\documentclass[3p]{elsarticle}
\usepackage[utf8]{inputenc}
\usepackage{amsmath}
\usepackage{amsfonts}
\usepackage{amssymb}
\usepackage{amsfonts}
\usepackage{amsthm}
\usepackage{amscd}
\usepackage{color}
\usepackage{graphics}
\usepackage{graphicx}
\usepackage[toc,page]{appendix}

\usepackage[all]{xy}

\usepackage{algorithm,algorithmic}

\usepackage{natbib}
\biboptions{sort&compress,numbers}

\newtheorem{thm}{\bf Theorem}[section]

\newtheorem{prp}[thm]{\bf Proposition}

\newtheorem{rem}[thm]{\bf Remark}

\newtheorem{exmp}[thm]{\bf Example}
\newtheorem{cor}[thm]{\bf Corollary}

\newcommand{\cis}[1][C]{\ensuremath{\mathbb{#1}}}
\newcommand{\av}[1]{\ensuremath{\mathcal{#1}}}

\newcommand{\vek}[1][h]{\ensuremath{\mathbf{#1}}}

\newcommand{\prR}[1]{\ensuremath{\mathbb{P}^{#1}_{\cis[R]}}}

\newcommand{\af}[1]{\ensuremath{\mathbb{A}^{#1}_{\cis}}}
\newcommand{\afR}[1]{\ensuremath{\mathbb{A}^{#1}_{\cis[R]}}}

\newcommand{\euR}[1]{\ensuremath{\mathbb{E}^{#1}_{\cis[R]}}}

\journal{Journal of Computational and Applied Mathematics}

\begin{document}

\sloppy

\begin{frontmatter}

\title{Determining surfaces of revolution from their implicit equations}

\author[plzen1]{Jan Vr\v{s}ek}
\ead{vrsek@ntis.zcu.cz}

\author[plzen2,plzen1]{Miroslav L\'avi\v{c}ka\corref{cor1}}
\cortext[cor1]{Corresponding author}
\ead{lavicka@kma.zcu.cz}

\address[plzen1]{NTIS -- New Technologies for the Information Society, Faculty of Applied Sciences, University of West Bohemia,
         Univerzitn\'i 8, 301 00 Plze\v{n}, Czech~Republic}

\address[plzen2]{Department of Mathematics, Faculty of Applied Sciences, University of West Bohemia,
         Univerzitn\'i~8,~301~00~Plze\v{n},~Czech~Republic}

\begin{abstract}
Results of number of  geometric operations (often used in technical practise,  as e.g. the operation of blending) are in many cases surfaces described implicitly. Then it is a challenging task to recognize the type of the obtained surface, find its characteristics and for the rational surfaces compute also their parameterizations. In~this contribution we will focus on surfaces of revolution. These objects, widely used in geometric modelling, are generated by rotating a generatrix around a given axis. If the generatrix is an algebraic curve then so is also the resulting surface, described uniquely by a polynomial which can be found by some well-established implicitation technique. However, starting from a polynomial it is not known how to decide if the corresponding algebraic surface is rotational or not. Motivated by this, our goal is to formulate a simple and efficient algorithm whose input is a~polynomial with the coefficients from some subfield of $\cis[R]$ and the output is the answer whether the shape is a surface of revolution. In the affirmative case we also find the equations of its axis and generatrix.
Furthermore, we investigate the problem of rationality and unirationality of surfaces of revolution and show  that this question can be efficiently answered discussing
the rationality of a certain associated planar curve.
\end{abstract}


\end{frontmatter}

%%%%%%%%%%%%%%%%%%%%%%%%%%%%%%%%%%%%%%%%%%%%%%%%%%%%%%%%%%%%%%%%%%%%%%%%%%%%%%%%%%%%%%%%%%%%%%%%%%%%%%%%%%%%%%%%%%%%%%%%%%%%%%%%%%%%%%%
\section{Introduction and related work}\label{Intro}
%%%%%%%%%%%%%%%%%%%%%%%%%%%%%%%%%%%%%%%%%%%%%%%%%%%%%%%%%%%%%%%%%%%%%%%%%%%%%%%%%%%%%%%%%%%%%%%%%%%%%%%%%%%%%%%%%%%%%%%%%%%%%%%%%%%%%%%

The choice of a suitable description of a given shape (parametric, or implicit) is a fundamental thing for designing and studying efficient subsequent geometric algorithms in many technical applications. Parameterizations, most often used in Computer-Aided (Geometric) Design, allow us to generate points on curves and surfaces, they are also very suitable for plotting, computing transformations, computing curvatures e.g. for shading and colouring etc. On the other hand implicit representations are especially suitable for deciding whether a given point is lying on the object, or outside. In addition, it is convenient to intersect two shapes when one is given parametrically and the other implicitly. Finally, in computer graphics, ray tracing is efficiently used for generating an image of implicit algebraic surfaces.

However, we must recall that not every algebraic curve or surface admits a rational parameterization. To be more exact, let $\av{X}$ be a variety  over a field $\mathbb{K}$. Then $\av{X}$ is said to be {\em unirational} if it admits a rational parameterization. Furthermore, if there exists a proper
parameterization (i.e., a parameterization with the rational inverse) then $\av{X}$ is called {\em rational}.
By the theorem of L\"uroth, a curve has a parameterization if and only if it has a proper parameterization if and only if its genus (see \cite{Wa50} for a definition of this notion) vanishes. Hence, for planar curves the notions of rationality and unirationality are equivalent for any field. Algorithmically, the parameterization problem is well-solved, see e.g. \cite{vH97,SeWi91,Wa50}.  In the surface case the theory differs. Over algebraically closed field with characteristic zero,
by the~Castelnuovo's theorem, surface is unirational if and only if it is rational if and only if the arithmetical genus $p_a$ and the second plurigenus $P_2$ are both zero (see \cite{Ha77} for a definition of these notions).  The problem is algorithmically much more difficult than for curves -- see e.g. \cite{Sc98} for further details.

The reverse problem (consider a rational parametric description of a curve or a surface, find the corresponding implicit equation) is called the implicitization problem. For any rational parametric curve or surface, we can always convert it into implicit form. Nonetheless, the implicitization always involves relatively complicated process and the resulting implicit form might have large number of coefficients -- so, it is not a simple task in general. One can find many generic methods for implicitizing arbitrary rational curves and surfaces such as resultants, Gr\"{o}bner bases, moving curves and surfaces, and $\mu$-bases -- see e.g. \cite{GVNePDSeSe04,Ko04,SeChe95,SeGoDu97}.

In what follows we will deal with implicit surfaces of revolution which are created by rotating a curve around a straight line. Revolution surfaces are well known since ancient times and very common objects in geometric modelling, as they can be found everywhere in nature, in human artifacts, in technical practise and also in mathematics. There has been a thorough previous investigation  on finding the implicit equation of a rational surface of revolution. In \cite{ShJu05}, the authors created an implicit representation for surfaces of revolution by eliminating the square root from $f(\sqrt{x^2 + y^2}, z)$, where $f(x,z) = 0$ is the implicit equation of the generatrix curve.  Another approach to implicitizing rational surfaces of revolution was presented in  \cite{Chi09} where the method of moving planes was efficiently used -- the implicit equation of the surface of revolution is then given by the determinant of the matrix whose entries are the $2n$ moving planes that follow the surface, each derived from a distinct $3\times 3$ determinant. A recent technique for implicitizing rational surfaces of revolution was presented in \cite{ShGo12}. In this paper, the  $\mu$-bases for all the moving planes that follow the surface of revolution were found and subsequently the resultants were used to construct the implicit equation.

In this paper, we will investigate a different challenging problem of computational geometry originated in technical practise. We start with an~implicit representation and our goal is to decide if the corresponding algebraic surface is rotational or not. Moreover, in case of the positive answer we also want to compute the equations of the axis and the  generatrix of the rotational surface. We would like to stress out that this study reflects the need of the real-world applications as the results of many geometric operations are often described only implicitly. Then it is a challenging task to recognize the type of the obtained surface, find its characteristics and for the rational surfaces compute also their parameterizations. Let us recall e.g. the implicit blend surfaces (often of the canal/pipe/rotational-surface type) offering a good flexibility for designing blends as their shape is not restricted to be constructed as an embedding of a parameter domain. Important contributions for blending by implicitly given surfaces can be found in \cite{HoHo85,Ro89}; several methods for constructing implicit blends were thoroughly investigated in \cite{Ha90,Ha01}. Obviously, for choosing a suitable consequent geometric technique is necessary to decide the exact type of the constructed surface. So, the main contribution of this paper is answering the question for the surfaces of revolution which is mentioned in \cite{AnReSeTaVi14} as still unsolved. In addition we will also focus on the question of rationality and unirationality of surfaces of revolution a show  that this problem can be efficiently solved transforming it to the question of rationality of a planar curve.

The rest of  the paper is organized as follows. In Section~\ref{sec recognition} we consider an algebraic surface given by equation $f(x,y,z)=0$ for an irreducible polynomial defined over some  subfield $\cis[K]$ of $\cis[R]$, typically $\cis[Q]$ or its algebraic extensions. The goal is to decide whether the surface is rotational and eventually to find its axis and profile curve. In this part a symbolic algorithm  for recognition of surfaces of revolution is designed and thoroughly discussed. Section~\ref{sec rationality} deals with the relation between the profile curve (and its quadrat) and the (uni)rationality of the associated surface of revolution. Properties of tubular surfaces are exploited to formulate the results about rationality of surfaces of revolution. Finally we conclude the paper in Section~\ref{Concl}. The theory is documented in detail on two computed examples presented in Appendix.

%%%%%%%%%%%%%%%%%%%%%%%%%%%%%%%%%%%%%%%%%%%%%%%%%%%%%%%%%%%%%%%%%%%%%%%%%%%%%%%%%%%%%%%%%%%%%%%%%%%%%%%%%%%%%%%%%%%%%%%%%%%%%%%%%%%%%%%
\section{Implicit surfaces of revolution and their recognition}\label{sec recognition}
%%%%%%%%%%%%%%%%%%%%%%%%%%%%%%%%%%%%%%%%%%%%%%%%%%%%%%%%%%%%%%%%%%%%%%%%%%%%%%%%%%%%%%%%%%%%%%%%%%%%%%%%%%%%%%%%%%%%%%%%%%%%%%%%%%%%%%%

Let be given a straight line $\av{A}$ in Euclidean space $\euR{3}$ and let  $\av{G}\subset\euR{3}$ be an  algebraic space curve distinct from $\av{A}$. We assume that $\av{G}$ is not a~line perpendicular to $\av{A}$. Then the object $\av{X}$ created by rotating $\av{G}$ around  $\av{A}$ is an~algebraic surface  which is called a~\emph{surface of revolution} (in what follows, we will write shortly SOR) with the \emph{axis} $\av{A}$ and the \emph{generatrix} $\av{G}$, see Fig.~\ref{fig SOR}~(Left). Assume $\av{X}$ is given by the equation $f(x,y,z)=0$ where $f\in\cis[K][x,y,z]$ for a field $\cis[K]$. In addition we consider that $\av{X}$ is absolutely irreducible (i.e., $f\not= f_1\cdot f_2$ for $f_1,f_2\in\cis[C][x,y,z]$).

Of course, there exist a lot of generating curves $\av{G}$ leading to the same surface. Among them we can find one with a~prominent role -- the~\emph{profile curve} $\av{P}$, i.e., the intersection of $\av{X}$ with a~plane containing the axis, see Fig.~\ref{fig SOR}~(Right).
\begin{figure}[t]
\begin{center}
   \hspace*{6ex}\raisebox{5ex}{\includegraphics[width=0.3\textwidth]{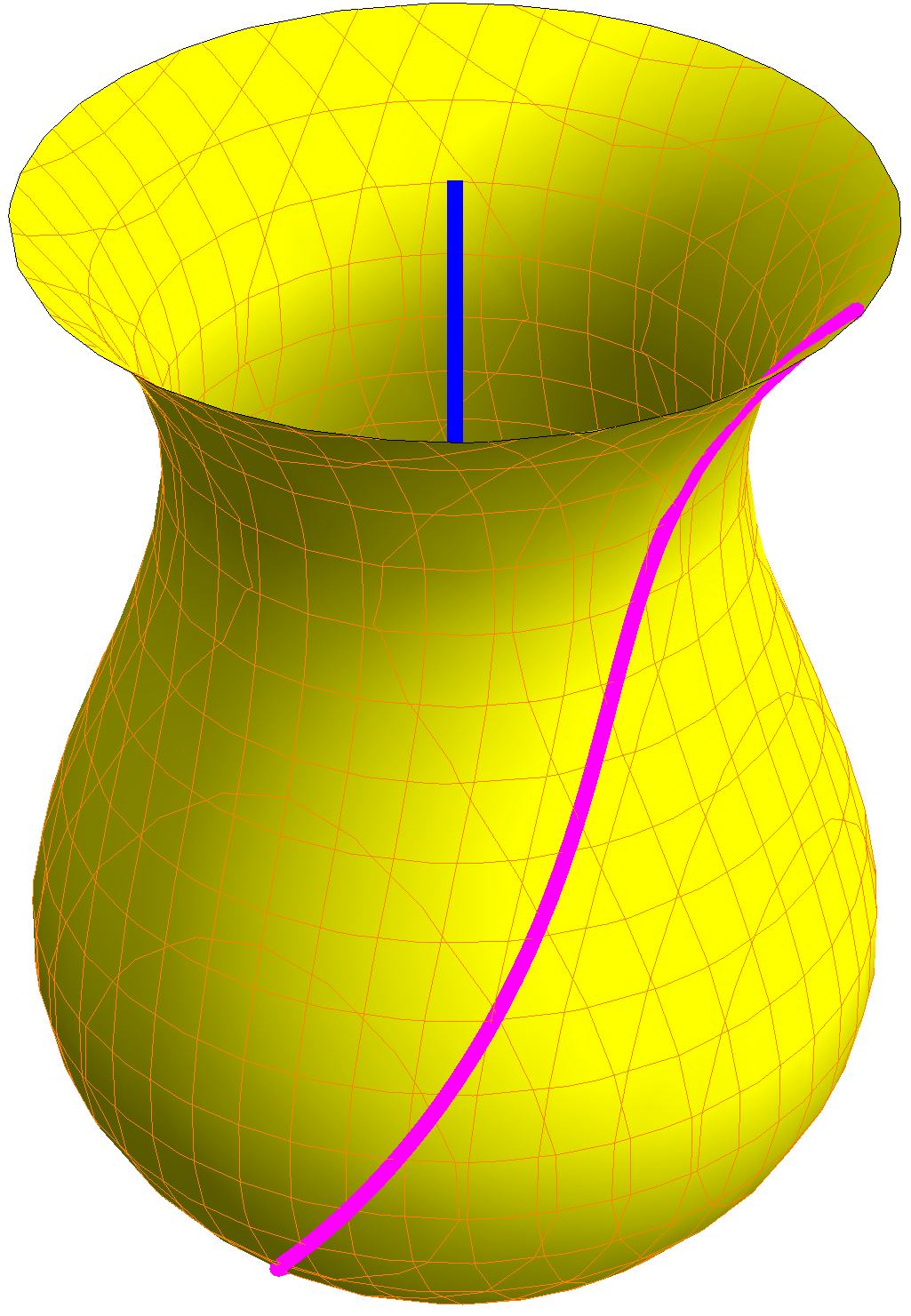}}\hfill
   \includegraphics[width=0.4\textwidth]{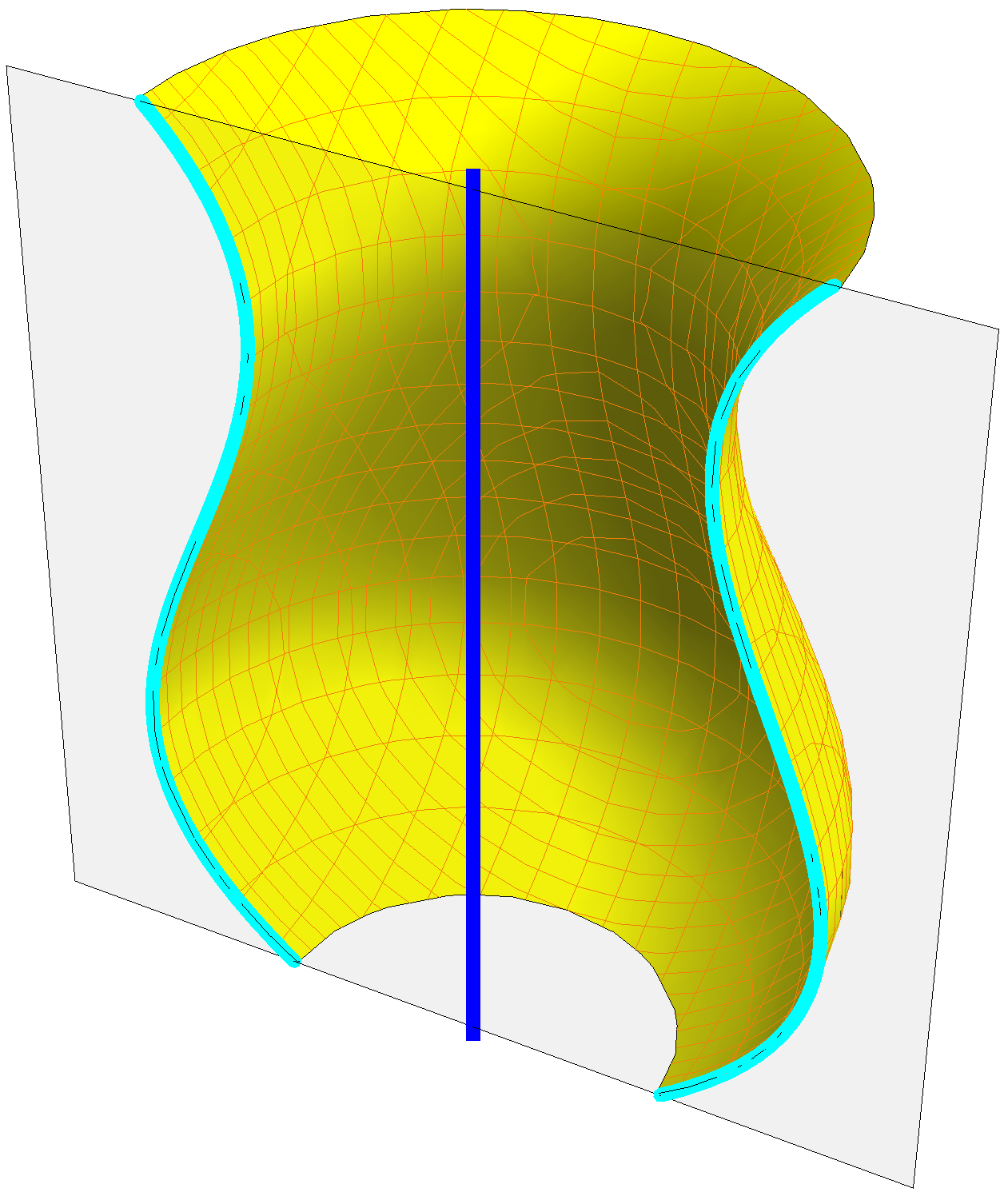}\hspace*{6ex}

\begin{minipage}{0.9\textwidth}
\caption{Left: A surface of revolution (yellow), its axis (blue) and the generatrix (magenta).
Right: A surface of revolution (yellow) cut by a plane (grey) containing the axis (blue) and the profile curve (cyan).
\label{fig SOR}}
\end{minipage}
\end{center}
\end{figure}

In this paper, we want to solve the problem of determining surfaces of revolution from their implicit equations.
Our goal is to formulate a simple and efficient symbolic algorithm whose input will be a~polynomial with the coefficients from a field and the output will be the decision whether the described algebraic surface is SOR or not. We start with a considerably simpler situation -- in particular, we assume that  $\av{X}$ is SOR whose axis coincides with the coordinate $x$-axis. Thus we may obtain its profile curve $\av{P}$ by intersecting $\av{X}$ for instance with the plane $z=0$. Hence, we can consider $\av{P}$ as a curve in $xy$-plane.  Obviously, $\av{P}$ is symmetric with respect to the $x$-axis. Since $(x,y)\in\av{P}$ if and only if $(x,-y)\in\av{P}$ we deduce that its equation $p(x,y)=f(x,y,0)=0$ can be written in one of the following two forms
\begin{equation}
  \sum_i p_i(x)y^{2i} = 0,\qquad\mathrm{or}\qquad \sum_i p_i(x)y^{2i+1}=0.
\end{equation}
Nevertheless the second polynomial can be factorized as $y(\sum_i p_i(x)y^{2i})=0$ which implies that SOR contains a degenerated component $y^2+z^2=0$. This is a contradiction with the~assumed absolute irreducibility of $f(x,y,z)$.  So in what follows we will work with the defining polynomial of $\av{P}$ only in the form
\begin{equation}\label{eq profile}
  p(x,y) = f(x,y,0)=\sum_i p_i(x)y^{2i},
\end{equation}
i.e., $p(x,y)$ contains $y$ solely in even powers. 

Nonetheless, despite $\av{X}$ being irreducible the profile curve may still be either irreducible, or it may decompose into two components. As an~example we can take the hyperboloid of revolution whose profile curve is the irreducible hyperbola. On the other hand the profile curve of the cone of revolution consists of the two intersecting lines (one is the reflected image of the other along the $x$-axis). And this is a~general rule -- if the profile curve of some irreducible SOR is reducible, i.e., $\av{P=P^+\cup P^-}$, then $\av{P}^+$ and $\av{P}^-$ are symmetrically conjugated with respect to the $x$-axis. Moreover, the defining polynomial of $\av{P}$ has then the form $p(x,y)=g(x,y)\cdot g(x,-y)$.

Conversely, starting with the~profile curve \eqref{eq profile}  the corresponding SOR consists of the~points $(x,y,z)$ such that $\left(x,\sqrt{y^2+z^2}\right)\in\av{P}$. Hence its defining polynomial can be written as $\sum_i p_i(x)(y^2+z^2)^i$. Let us summarize this observation to the following proposition:

\begin{prp}\label{prp equation of sor}
 $\av{X}$ is a surface of revolution with the axis $x$ and the profile curve $p(x,y)=0$ if and only if $f(x,y,z)=p(x,y^2+z^2)$.
\end{prp}

\smallskip
If the axis $\av{A}$ is in a general position we can find a suitable isometry $\phi$ which maps it to the coordinate axis $x$ and then we test the transformed surface on the SOR property. Therefore the recognition whether $\av{X}$ is SOR can be reduced only to finding the axis. The identification of the axis will be based on the following well known property of the normal lines to SOR.

\begin{prp}\label{prp known property}
Let $\av{X}$ be a surface of revolution with the axis $\av{A}$. Then the normal line through its non-singular point intersects $\av{A}$ or it is parallel to $\av{A}$.
\end{prp}

\begin{algorithm}[t]
\caption{Recognition of SOR I}\label{alg recognition I} \algsetup{indent=2em}
\begin{algorithmic}[1]
 \REQUIRE $\av{X}:\ f(x,y,z)$
 \ENSURE \av{X} is SOR with the axis {\av{A}}/\av{X} is not SOR.

 \STATE Find sufficiently enough points $\{\vek[p]_1,\dots,\vek[p]_n\}$ points on $\av{X}$.
 \STATE Compute the normals $\{N_{\vek[p]_1}\av{X},\dots,N_{\vek[p]_n}\av{X}\}$ of $\av{X}$ at $\vek[p]_i$.
 \IF{$\exists!$ a straight line $\av{A}$ s.t. $\forall i:\ \av{A}\cap N_{\vek[p]_i}\av{X}\not=\emptyset$}
   \STATE Find the isometry $\phi$ s.t.  $\phi(\av{A})= \langle x\rangle$.
   \STATE $\av{X}' = \phi(\av{X})$
   \IF{$\av{X}'$ is SOR with the $x$-axis}
     \RETURN \av{X} is SOR with the axis {\av{A}}.
   \ELSE
     \RETURN \av{X} is not SOR.
   \ENDIF
 \ELSE
   \RETURN \av{X} is not SOR.
 \ENDIF
\end{algorithmic}
\end{algorithm}

Following the discussion given above, we formulate the first naive algorithm, see Algorithm~\ref{alg recognition I}.
However, this~algorithm has some~serious gaps. In particular we do not know in step 2  how many points are enough and how to find them. Later, using Pl\"{u}cker coordinates, we will see that 5 points in general position are sufficient. Nonetheless, the next drawback is even more serious. Note that finding points on the implicit surface leads to solving polynomial equations, which is computationally hard task for surfaces of higher degree.  In addition, we emphasize that the purpose of the algorithm (i.e., its symbolic character) does not allow us to use numerical approximations only -- see for instance the following example where numerical computations lead to a wrong conclusion.

\begin{exmp}\rm\label{exmp numeric}
  Let $\av{X}:\ y^2-2xz=0$ be a~cone of revolution with the~axis $(t,0,t)$. We assume that due to computer computations the~axis is obtained not exactly but with some perturbed float coefficients, e.g. $(t,0,1.0000001t)$. Then the~transformed surface $\av{X}'$  possesses the~equation $f(x,y,z)=y^2+z^2-0.0000002xz -x^2=0$. Hence the~profile curve $p(x,y)=0$ should be given by $p(x,y)=y^2-x^2=0$. Since it contains $y$ in even powers only, it is a~profile curve of some SOR. However, $p(x,\sqrt{y^2+z^2})=y^2+z^2-x^2\not\sim f(x,y,z)$ and thus $\av{X}'$ is not SOR with the~computed axis and the~algorithm fails.
\end{exmp}

\medskip
Before we formulate a new version of the recognition algorithm we will study an interesting property of surfaces of revolution which will help us to avoid numerical computations needed in   Algorithm~\ref{alg recognition I}.

Consider a polynomial $f\in\cis[R][x,y,z]=\cis[R][\vek[x]]$ and let $\av{X}:f(\vek[x])=0$  be an algebraic surface in Euclidean space $\euR{3}$. Then by $\av{X}_\alpha$ ($\alpha\in\cis[R]$) we will denote a surface with the defining equation $f(\vek[x])=\alpha$. The 1--parametric family of such surfaces is denoted by $\Sigma_f$. Since the value $f(\vek[p])$ is well defined for each point $\vek[p]\in\euR{3}$ we can see that $\vek[p]\in\av{X}_{f(\vek[p])}\in\Sigma_f$. Hence through each point of $\euR{3}$ passes exactly one surface from  the family $\Sigma_f$. This is a distinguished property of SORs that will play a crucial role in our recognition algorithm.
\begin{figure}[t]
\begin{center}
   \hspace*{6ex}\raisebox{0ex}{\includegraphics[width=0.35\textwidth]{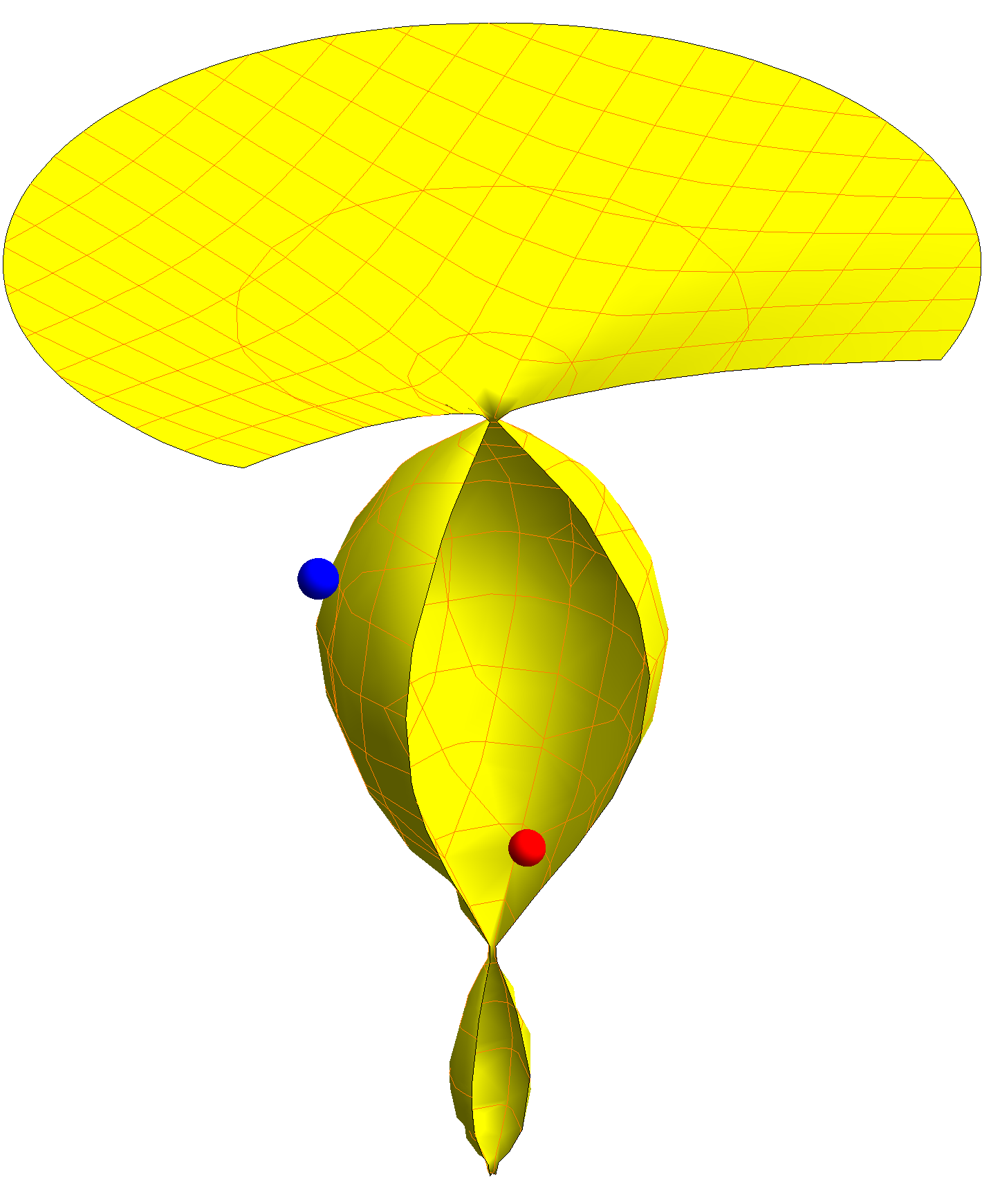}}\hfill
   \includegraphics[width=0.4\textwidth]{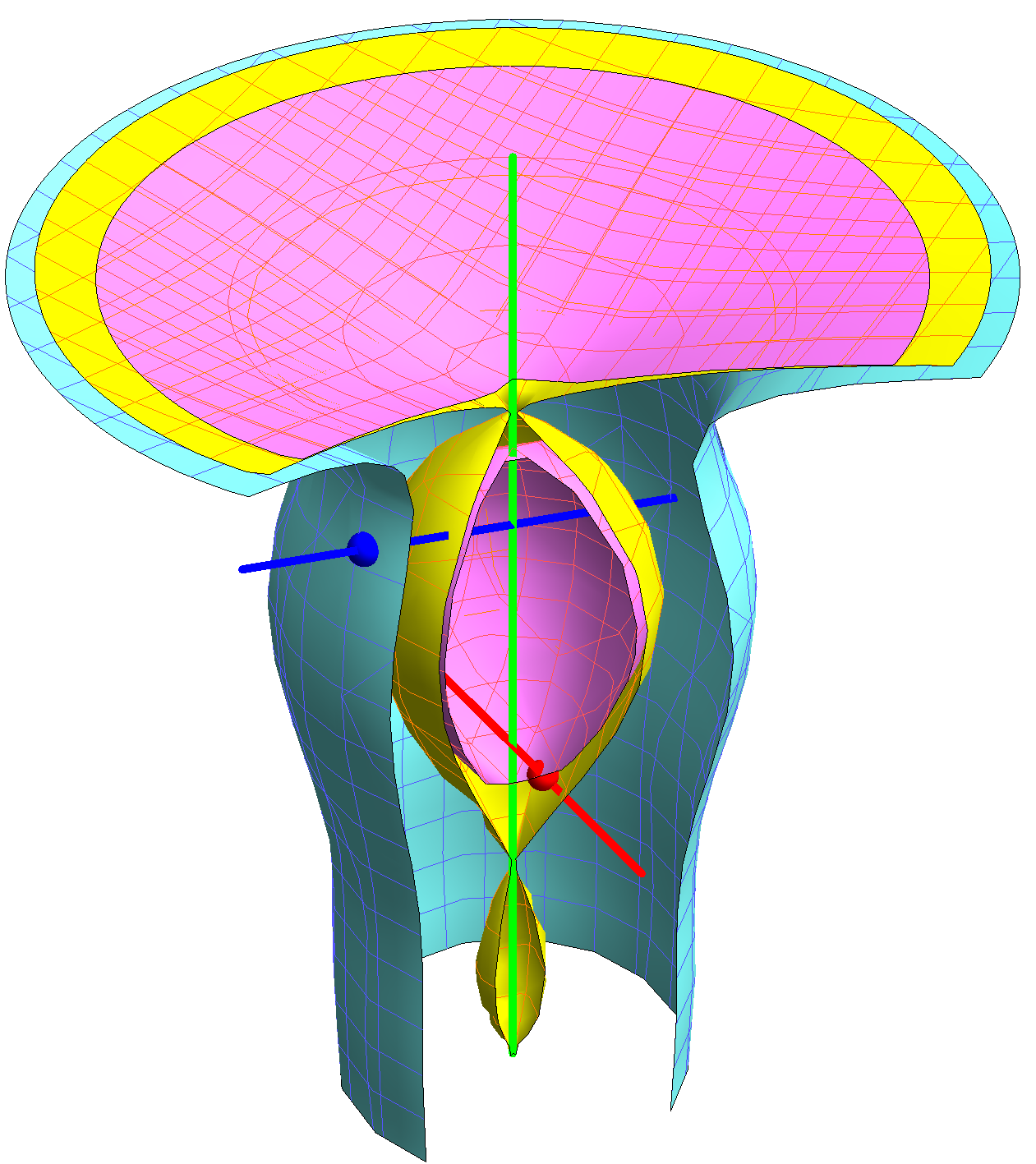}\hspace*{6ex}

\begin{minipage}{0.9\textwidth}
\caption{Left: A surface of revolution $f(x,y,z)=0$ (yellow) and two points (red, blue) arbitrarily chosen in space. Right: Two surfaces of revolution (magenta, cyan) from the family $\Sigma_f$ uniquely determined by the two chosen points  and the common normals (red, blue) of all surfaces intersecting the axis of rotation (green).
\label{fig onion}}
\end{minipage}
\end{center}
\end{figure}

\begin{thm}\label{thm sor are onions}
  Let $\av{X}$ be a surface of revolution with the axis $\av{A}$. Then for any $\alpha\in\cis[R]$ the surface $\av{X}_\alpha$ is also a surface of revolution with the axis $\av{A}$.
\end{thm}

\begin{rem}\rm
  Let $\av{X}$ be given by the polynomial $x^2+y^2-1$, i.e., it is a cylinder of revolution with the axis $z$. Then the real part of $\av{X}_{-1}$ is only the coordinate $z$-axis and for each $\alpha<-1$ it is the empty set. Nonetheless, it is still possible to consider complex surfaces of revolution to overcome this limitation. However this is not needed in the presented paper.
\end{rem}

\begin{proof}[Proof of Theorem~\ref{thm sor are onions}]
  First we prove the theorem for a special case when the axis of SOR $\av{X'}$ coincides with the coordinate $x$-axis. By Proposition~\ref{prp equation of sor}, we know that $f(x,y,z)=\sum_i p_i(x)(y^2+z^2)^i$. Then $\av{X'}_\alpha$ is defined by the polynomial
  \begin{equation}
    f(x,y,z)-\alpha=  (p_0(x)-\alpha)+\sum_{i\geq 1}p_i(y^2+z^2)^i,
  \end{equation}
  which is obviously again the equation of some SOR with the same axis.

  To complete the proof suppose that $\av{X}$ is SOR in generic position and $\phi:\euR{3}\rightarrow\euR{3}$ is an isometry mapping axis of $\av{X}$ to the $x$-axis, i.e., $\phi(\av{X})=\av{X'}$. Since each $\phi$ can be written as $\vek[x\mapsto A\cdot x+b]$  for $\vek[A]\in\mathrm{O}_3(\cis[R])$ and $\vek[b]\in\cis[R]^3$, we obtain the~defining polynomial of $\av{X}'$ in the form $f(\vek[A]^{-1}\cdot(\vek[x-b]))$. So the surface $\phi(\av{X}_\alpha)$ admits the~equation
\begin{equation}
   f((\vek[A]^{-1}\cdot(\vek[x-b]))=\alpha
\end{equation}
and thus $\phi(\av{X}_\alpha) = \av{X}'_\alpha$ is  SOR by the arguments from the beginning of this proof.
\end{proof}

As a significant  practical contribution of the previous theorem we do not need to calculate the points $\{\vek[p]_i\}$ on $\av{X}$ (cf. Algorithm~\ref{alg recognition I}) but it is sufficient to choose them anywhere in $\euR{3}$ and then find a straight line $\av{A}$ intersecting all the normals $N_{\vek[p]_i}\av{X}_{f(\vek[p]_i)}$, see Fig~\ref{fig onion}. The~best tools for such kind of computations offers the line geometry, for the introduction to this branch of geometry see e.g. \cite{PoWa01}.  Recall that to each line $\av{L}$ in $\euR{3}$ determined by a point $\vek[p]$ and a direction vector $\vek[v]$ we may associate a homogeneous six-tuple
\begin{equation}
  L=(l_1:l_2:\cdots:l_6)=(\vek[l]:\overline{\vek[l]})=(\vek[v]:\vek[p\times v])\in\prR{5},
\end{equation}
the so called  \emph{Pl\"{u}cker coordinates}. With a bilinear form
\begin{equation}
  \langle X,Y\rangle =\langle (\vek[x:\overline{x}), (y:\overline{y})\rangle=x\cdot\overline{y}+\overline{x}\cdot y],
\end{equation}
where `$\cdot$' is a standard Euclidean inner product, we have that ({\it i}) $X$ represents a line if and only if $\langle X,X\rangle = 0$, and ({\it ii}) lines $X$ and $Y$ intersect (or they are parallel) if and only if $\langle X,Y\rangle = 0$.

Now, let us consider the Pl\"{u}cker coordinates of the normal $N_{\vek[p]_i}\av{X}_{f(\vek[p]_i)}$
\begin{equation}
  N_i=(\vek[n]_i:\overline{\vek[n]}_i)=(\nabla\,f(\vek[p]_i):\vek[p]_i\times\nabla\,f(\vek[p]_i)).
\end{equation}
If $A=(\vek[a]:\overline{\vek[a]})$ are the Pl\"{u}cker coordinates of the (sought) axis of SOR then the geometric condition that $\av{A}$ intersects all the normals reads as
\begin{equation}\label{eq_soustava}
    \vek[n]_i\cdot\overline{\vek[a]}+\overline{\vek[n]}_i\cdot\vek[a]=0,
\end{equation}
which is a system of homogeneous linear equations in six variables. Hence, it is enough to consider only five linearly independent normals $N_i$ to compute $A$. This brings us to an improved version of Algorithm~\ref{alg recognition I} without drawbacks discussed before -- see Algorithm~\ref{alg recognition II}.

\begin{algorithm}[h]
\caption{Recognition of SOR II}\label{alg recognition II} \algsetup{indent=2em}
\begin{algorithmic}[1]
 \REQUIRE $\av{X}:\ f(x,y,z)$
 \ENSURE \av{X} is SOR with the axis {\av{A}}/\av{X} is not SOR.

 \STATE $N:=\{\}$; $i:=1$
 \WHILE{$i<5$}
   \STATE Choose a random point $\vek[p]_i$
   \STATE $N_i:= (\nabla\,f(\vek[p]_i):\vek[p]_i\times\nabla\,f(\vek[p]_i))$
   \IF{$\{N_1,\dots,N_i\}$ are linearly independent}
     \STATE $N:=N\cup N_i$; $i:=i+1$
   \ENDIF
 \ENDWHILE
 \STATE Solve the system $\vek[n]_j\cdot\overline{\vek[a]}+\overline{\vek[n]}_j\cdot\vek[a]=0$ for $j=1,\dots,5$.
 \IF{$(\vek[a],\overline{\vek[a]})$ determines a straight line}
   \STATE $\av{A}$ is the line with the coordinates $(\vek[a]:\overline{\vek[a]})$
   \STATE Find the isometry $\phi$ s.t.  $\phi(\av{A})= \langle x\rangle$.
   \STATE $\av{X}' = \phi(\av{X})$
   \IF{$\av{X}'$ is SOR with the $x$-axis}
     \RETURN \av{X} is SOR with the axis {\av{A}}.
   \ELSE
     \RETURN \av{X} is not SOR.
   \ENDIF
 \ELSE
   \RETURN \av{X} is not SOR.
 \ENDIF
\end{algorithmic}
\end{algorithm}

Moreover if $f$ is a~polynomial from $\cis[K][x,y,z]$ then taking points with their coordinates from $\cis[K]^3$ leads to a system of linear equations with the coefficients in $\cis[K]$. Thus, the following corollary easily follows:

\begin{cor}\label{cor k-rational axis}
   Let the surface of revolution be given by the equation with the coefficients in a~field $\cis[K]$. Then its axis admits a~parameterization with the coefficients also in $\cis[K]$.
\end{cor}

%%%%%%%%%%%%%%%%%%%%%%%%%%%%%%%%%%%%%%%%%%%%%%%%%%%%%%%%%%%%%%%%%%%%%%%%%%%%%%%%%%%%%%%%%%%%%%%%%%%%%%%%%%%%%%%%%%%%%%%%%%%%%%%%%%%%%%%
\section{Rationality of surfaces of revolution}\label{sec rationality}\label{sec rational}
%%%%%%%%%%%%%%%%%%%%%%%%%%%%%%%%%%%%%%%%%%%%%%%%%%%%%%%%%%%%%%%%%%%%%%%%%%%%%%%%%%%%%%%%%%%%%%%%%%%%%%%%%%%%%%%%%%%%%%%%%%%%%%%%%%%%%%%

In the previous section we presented a simple and efficient method for recognition of implicitly given surfaces of revolution~ $\av{X}$. In the affirmative case we also obtained the profile curve $\av{P}: p(x,y)=0$ and the axis $\av{A}$.  However, in many (especially technical) applications it is usually more convenient to work with parametric representations of surfaces instead of with implicit ones. In this section we will focus on this problem and discuss the question of rationality of SORs. In what follows, we assume that the profile curve is a~real curve and thus $\av{X}$ is also real, i.e, it is a~two-dimensional subset of $\euR{3}$.

%Let us remind known results on rationality of surfaces first. An algebraic variety is said to be \emph{unirational} if it admits a rational parameterization. If there exists a proper (i.e. birational) parameterization then the variety is called \emph{rational}. Trivially every rational variety is unirational. The converse statement is also true for curves defined over an arbitrary field (L\"{u}roth theorem) and for surfaces defined over algebraically closed field of characteristic zero (Castelnuovo theorem), however it is not valid for real surfaces any more.

For a real algebraic surface $\av{X}\subset\afR{3}$, let $\av{X}_{\cis}$ be the surface defined by the same polynomial but considered in the complex space $\af{3}$.  Then it may be easily seen that for a~rational SOR the circles of latitude (characteristic circles) form a rational pencil on $\av{X}_{\cis}$ and hence this surface  is by N\"{o}ther theorem (see \cite{No70}) birationally equivalent to a~tubular surface. We recall that \emph{tubular surfaces} are shapes described by the equation
\begin{equation}
  A(x)y^2+B(x)z^2+C(x)=0.
\end{equation}
It is proved in \cite[Theorem 3]{Sc98b} that any real tubular surface is unirational. Hence there exists a rational parameterization if and only if $\av{X}$ is a~real surface and $\av{X}_{\cis}$ is rational. The rationality of complex surfaces may be tested via computing two of its birational invariants -- in particular by Castelnuovo's theorem it holds, $\av{X}_{\cis}$ is rational if and only if $P_2=p_a=0$, see e.g. Section~\ref{Intro}. Nevertheless the computation of these invariants is exceedingly complicated, in general. Hence we will use the fact that the surface is SOR and prove a criterion based on the rationality of a certain curve easily derived from the profile curve of the surface, which will be significantly a simpler problem.

Recall that, contrary to the curve case, the unirationality of a real surface does not imply its rationality. Nevertheless by Comesatti theorem (see \cite{Com12}) $\av{X}$ is rational if and only if it is unirational and connected (note that the number of components has to be computed in the projective extension and after resolving singularities). Since the number of components and the construction of rational parameterizations of tubular surfaces was thoroughly studied by Schicho, see e.g. \cite{Sc98b,Sch00}, it is sufficient to provide a~criterion of rationality of $\av{X}_{\cis}$ and to present explicitly a birational mapping from $\av{X}$ to a tubular surface.

\medskip
From now on, we will assume without loss of generality that the axis of SOR $\av{X}$ is the coordinate $x$-axis.
Trivially, if $(\phi(t),\psi(t))$ is a rational parameterization of the profile curve or one of its component then
\begin{equation}
  \vek[x](s,t)=\left(\phi(t),\frac{2s}{1+s^2}\psi(t),\frac{1-s^2}{1+s^2}\psi(t)\right)
\end{equation}
obviously parameterizes $\av{X}$.  However, this parameterization is not necessarily proper. Hence the rationality of $\av{P}$ implies the rationality of $\av{X}_{\cis}$ and at least the unirationality of $\av{X}$.  We emphasize that the converse statement is not true, i.e., there exists a rational SOR $\av{X}_{\cis}$ with the non-rational profile curve.

\begin{exmp}\rm
  The surface  $\av{X}:\ y^2+z^2=x^3+3x^2-2x$ is a cubic SOR. Its profile curve given by $ y^2=x^3+3x^2-2x$ is a non-singular cubic and hence it is an~elliptic curve. Nevertheless $\av{X}_{\cis}$ is a~rational surface as it is easy to verify that the curve $\av{G}$ parameterized by
\begin{equation}
  \left(\frac{t \left(\sqrt{2}t^2-1\right)}{\left(t^2+1\right)^2},-\frac{\left(1+\sqrt{2}\right)t^2}  {\left(t^2+1\right)^2},\frac{2t^4+3t^2+1}{(t^2+1)^2}\right)
\end{equation}
is a rational curve lying on the surface with the non-constant $x$-coordinate. Thus rotating this curve along the $x$-axis yields a rational parameterization of $\av{X}$, see Fig.~\ref{fig section}. One can also see in this figure  that the real part of $\av{G}$ does not intersect all of the real characteristic circles on $\av{X}$ and thus the obtained parameterization will not cover the whole surface but only one of its components. Moreover, the parameterized component $\av{G}$ intersects almost all the characteristic circles in two distinct points and thus after rotating the curve and generating SOR the obtained parameterization is non-proper. Indeed, $\av{X}$ is an example of real surfaces which are unirational but not rational because it consists of two connected components, see Fig~\ref{fig section}.

\begin{figure}[t]
\begin{center}
  \includegraphics[width=0.4\textwidth]{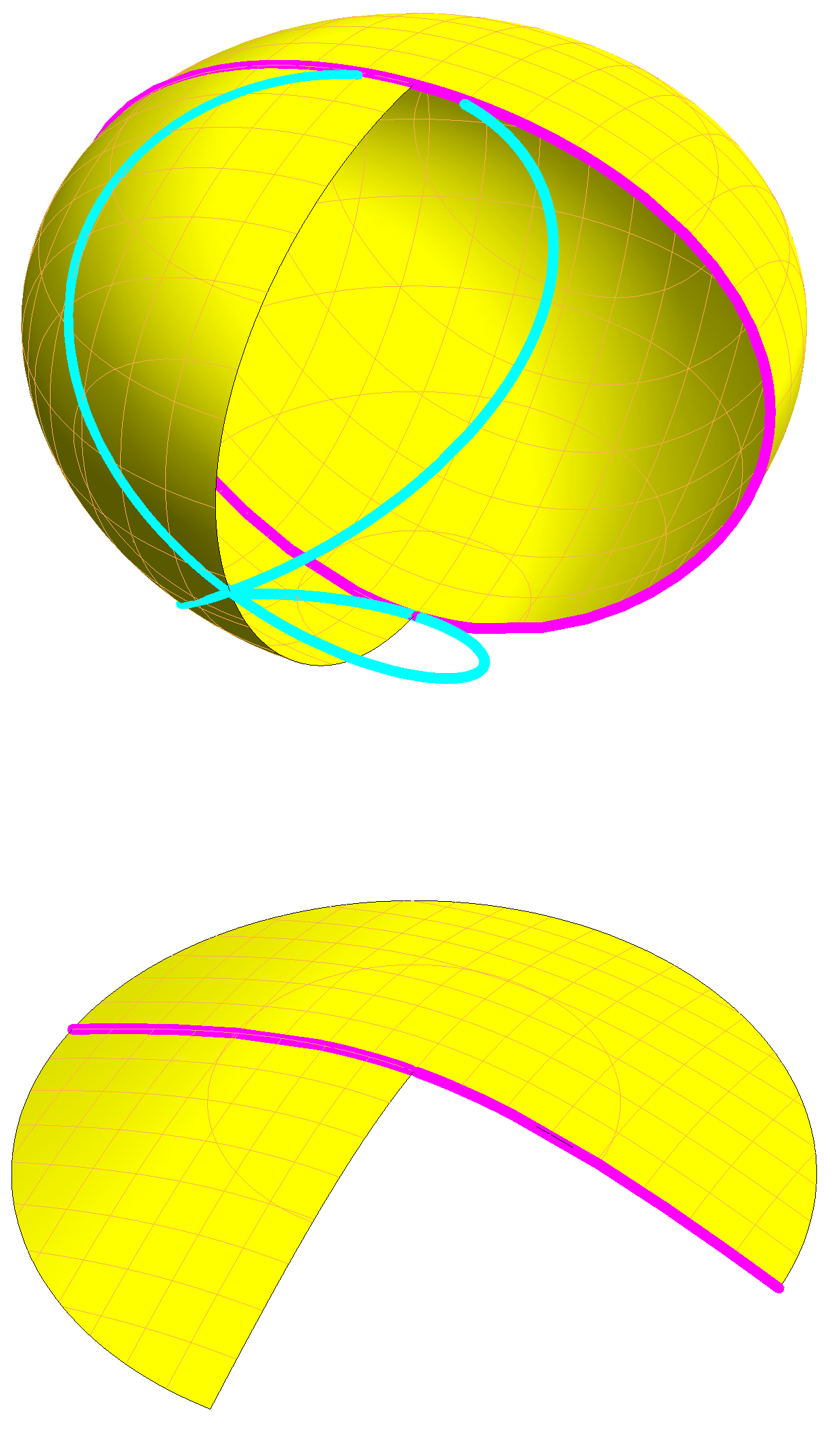}

\begin{minipage}{0.9\textwidth}
\caption{An example of unirational real SOR $\av{X}$ (yellow) consisting of the two components and with the elliptic, i.e., non-rational, profile curve $\av{P}$ (magenta) and a rational generatrix curve $\av{G}$ (cyan). The generatrix $\av{G}$ does not intersect all of the real circles of latitude on~$\av{X}$ and moreover it intersects the circles typically in two distinct points.
\label{fig section}}
\end{minipage}
\end{center}
\end{figure}
\end{exmp}

The curve used in the previous example represents objects which are fundamental for parameterizing surfaces of revolution. A rational curve on $\av{X}_{\cis}$ which intersects all the circles of latitude is called a \emph{section} of SOR. Having a parameterization $(\phi(s),\psi(s),\mu(s))$ of a~section of SOR one obtains a rational parameterization of $\av{X}_{\cis}$ simply by rotating the section along the axis, i.e., we arrive at
\begin{equation}
  \vek[x](s,t):=\left(\phi(t),\frac{2s}{1+s^2}\psi(t)-\frac{1-s^2}{1+s^2}\mu(t),\frac{2s}{1+s^2}\mu(t)+\frac{1-s^2}{1+s^2}\psi(t)\right)
\end{equation}
Conversely, a rational parameterization of SOR allows to generate a parameterization of the section in a straightforward way. Hence the rationality of $\av{X}_{\cis}$ (and thus the unirationality of $\av{X}$) is equivalent to the~existence of a~section.

Consider the morphism $\af{2}\rightarrow\af{2}$ given by $(x,y)\mapsto (x,y^2)$ and denote by $\av{P}^2$ the image of the profile curve under this morphism. According to \eqref{eq profile} the curve $\av{P}^2$ has the equation
\begin{equation}
  \sum_i p_i(x)y^i=0
\end{equation}
and the following theorem holds.

\begin{thm}\label{thm rationality of sor}
  Let $\av{X}$ be SOR as above. Then $\av{P}^2$ is irreducible curve and it is rational if and only if $\av{X}_{\cis}$ is a rational surface.
\end{thm}
\begin{proof}
If $\av{P}$ is irreducible then $\av{P}^2$ is the~image of the irreducible curve under the morphism and thus it is irreducible. If $\av{P=P^+\cup P^-}$ and $(x_0,y_0)$ is a point on $\av{P}^2$ with $y_0\not=0$ then the preimage consists of the two points $(x_0,\pm\sqrt{y_0})$; one on $\av{P^+}$ and one on $\av{P}^-$. Thus the morphism glues the two components together and $\av{P}^2$ is irreducible.

Next, let $\av{X}_{\cis}$ be rational, i.e., there exists a rational section parameterized by $(\mu(s),\phi(s),\psi(s))$. Then $(\mu(s),\phi^2(s)+\psi^2(s))$ parameterizes $\av{P}^2_{\cis}$. Since $\av{P}$ is a~real curve by assumption, so is $\av{P}^2$ and thus it is rational by L\"{u}roth theorem.

Finally from the~parameterization $(\phi(s),\psi(s))$ of $\av{P}^2$ it is possible to obtain a parameterization of a~section just by writing $\psi(s)$ as a sum of two squares, e.g.
\begin{equation}\label{eq complex sos}
  \left(\phi(s),\frac{1}{2}(\psi(s)+1),\frac{1}{2\sqrt{-1}}(\psi(s)-1)\right).
\end{equation}
\end{proof}

\begin{cor}
  If $\av{X}$ is SOR with the reducible profile curve then $\av{X}_{\cis}$ is rational if and only if $\av{P}^\pm$ are rational.
\end{cor}

\begin{proof}
 If $\av{P=P^+\cup P^-}$ then $(x,y)\mapsto(x,y^2)$ defines a birational morphism $\av{P}^+\rightarrow\av{P}^2$. Hence by Theorem~\ref{thm rationality of sor}
$\av{X}_{\cis}$ is rational if and only if $\av{P}^+$ is (and so is $\av{P}^-$).
\end{proof}

The expression of parameterization \eqref{eq complex sos} of section in the proof of Theorem~\ref{thm rationality of sor} is based on the fact that any rational function over the complex field can be written as a sum of two squares. The decomposition over $\cis[R]$ (or its subfield $\cis[K]$) is more delicate, see eg. \citep{LaSchWiHi00,LaSchWi01}. We are not going to repeat these results here as well as we do not present a method to (properly) parameterize $\av{X}$. Our goal was to derive a criterion of unirationality of $\av{X}$. Instead of parameterizing the surface directly we just provide explicitly a birational mapping from a certain tubular surface to the given SOR. The methods of proper parameterizations of tubular surfaces can be found in \citep{Sc98b,Sch00}.

\begin{thm}\label{thm birational to tubular}
  Let $(p(t)/q(t),r(t)/q(t))$ be a~proper parameterization of $\av{P}^2$ and let $\hat{r}$ and $\hat{q}$ be a square-free parts of $r$ and $q$ such that $\gcd(\hat{r},\hat{q})=1$. Then $\av{X}$ is birationally equivalent to the tubular   surface
  \begin{equation}\label{eq tubular surface}
     \av{T}: y^2+z^2-\hat{r}(x)\hat{q}(x)=0.
  \end{equation}
Moreover the mapping $\tau:\av{T\dashrightarrow X}$ can be given explicitly, see \eqref{eq tau}.
\end{thm}

\begin{proof}
  Since the parameterization of $\av{P}^2$ is assumed to be proper there exists its rational inverse $(x,y)\mapsto\varphi(x,y)$. Now consider the rational mapping
  \begin{equation}
    \tau_1:(x,y,z)\mapsto (\varphi(x,y^2+z^2),y,z).
  \end{equation}
  Let $\av{T}'$ denote the image of $\av{X}$ under $\tau_1$ then it is easily verified that it admits an equation
  \begin{equation}
   \frac{1}{\gcd(r(x),q(x))}(q(x)(y^2+z^2)-r(x))=\tilde{q}(x)(y^2+z^2)-\tilde{r}(x)=0.
  \end{equation}
  Moreover $\tau_1$ is birational as its inverse is given simply by $(x,y,z)\mapsto (\tilde{p}(x)/\tilde{q}(x),y,z)$. To construct $\av{T}'\dashrightarrow\av{T}$ we may proceed as in the proof of Lemma 2 in \citep{Sc98b}, i.e., if $\tilde{r}\cdot \tilde{q}=\hat{r}\cdot \hat{q}\cdot d^2$, then the birational mapping
  \begin{equation}
    \tau_2: (x,y,z)\mapsto \left(x,\frac{\tilde{q}(x)y}{d(x)},\frac{\tilde{q}(x)z}{d(x)}\right)
  \end{equation}
  maps $\av{T}'$ to $\av{T}$ given by \eqref{eq tubular surface}. Hence $\tau^{-1}=\tau_2\circ\tau_1:\av{X}\dashrightarrow\av{T}$ is a birational mapping with the desired inverse $\tau:\av{T\dashrightarrow X}$ given by
  \begin{equation}\label{eq tau}
    (x,y,z)\mapsto \left(\frac{\tilde{p}(x)}{\tilde{q}(x)},\frac{d(x)y}{\tilde{q}(x)},\frac{d(x)z}{\tilde{q}(x)}\right).
  \end{equation}
\end{proof}

\begin{figure}
\begin{center}
\hspace{2ex}
\includegraphics[width=0.4\textwidth]{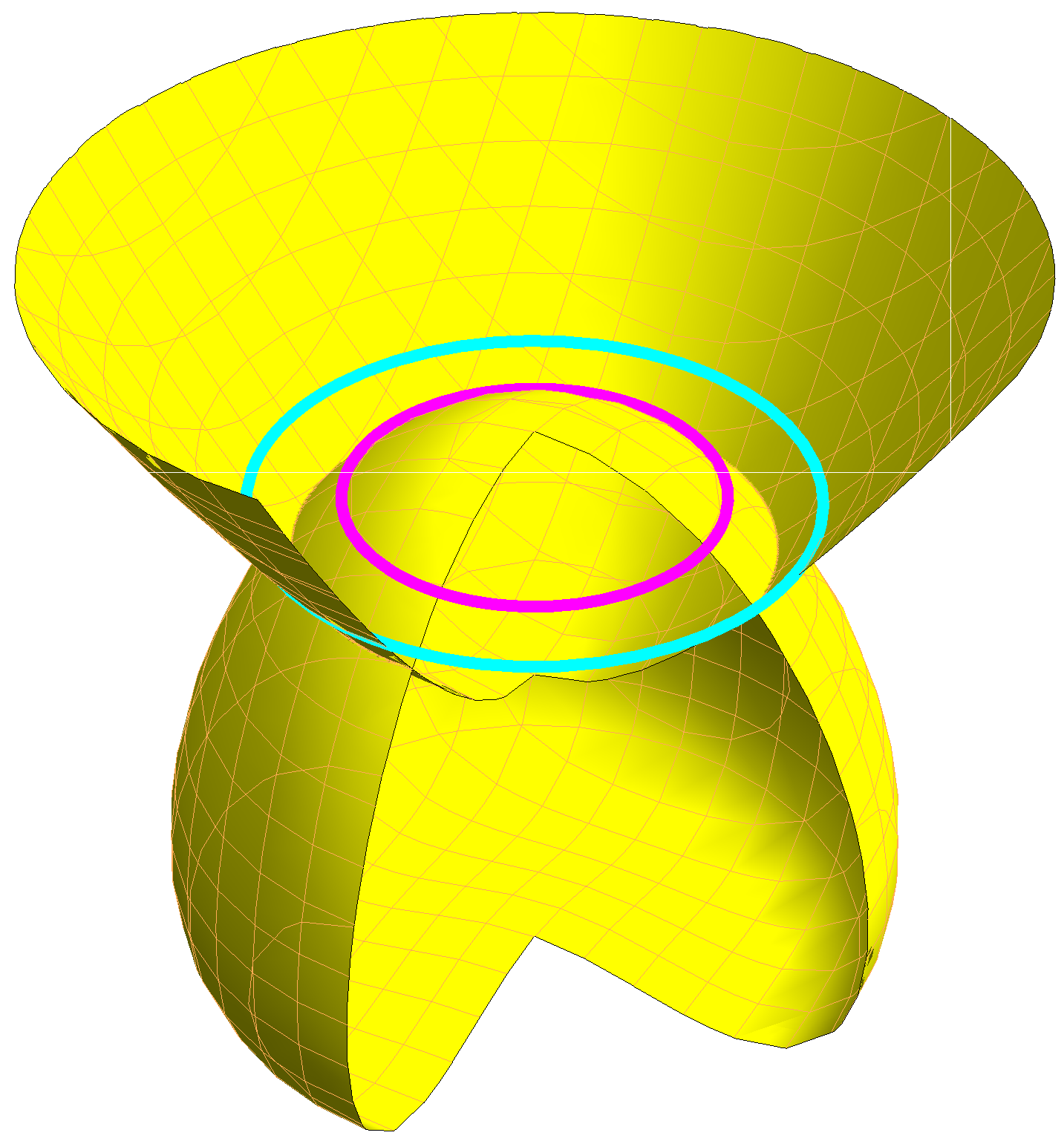}
\hfill
\raisebox{15ex}{\scalebox{1.2}{$\xrightarrow[\tau^{-1}]{\mathrm{tubularization}}$}}
\hfill
\raisebox{4ex}{\includegraphics[width=0.37\textwidth]{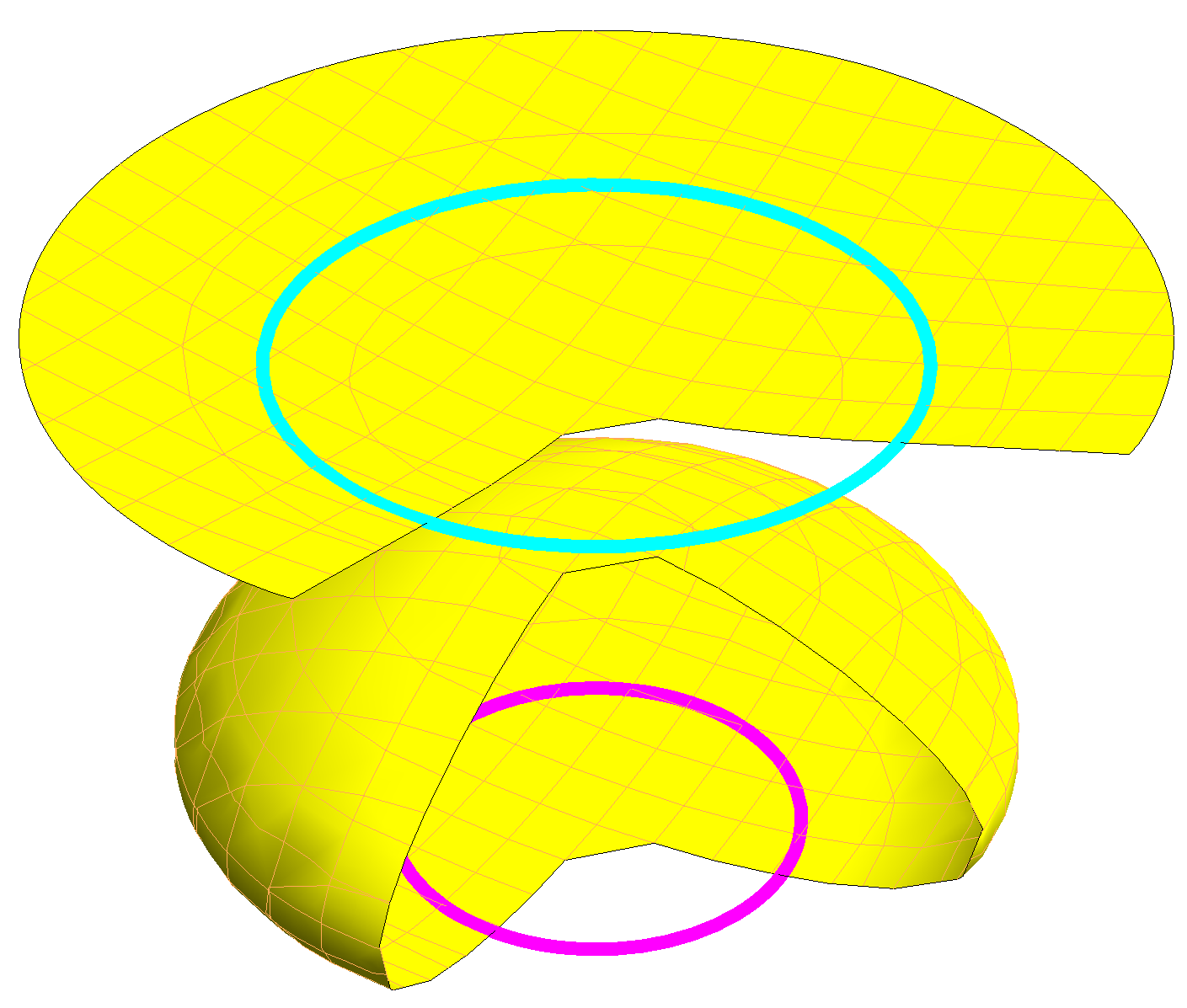}}
\hspace{2ex}
\begin{minipage}{0.9\textwidth}
\caption{Although a general SOR can have more circles with the same $x$ coordinate (i.e., co-centric circles of latitude), the corresponding tubular surface has only one circle with a given latitude. So the mapping $\tau^{-1}$ provides a separation of the co-centric circles.}
\label{fig TUB}
\end{minipage}
\end{center}
\end{figure}

\begin{cor}
  A unirational surface of revolution  $\av{X}$ is rational if and only if $\hat{r}(x)\hat{q}(x)$ has at most two real roots.
\end{cor}

\begin{proof}
 Applying the results on the spine curve of the~tubular from \cite{Sch00}, one can see that the number of connected components of $\av{T}$ equals to the number of intervals where the polynomial $\hat{r}(x)\hat{q}(x)>0$. Hence $\av{T}$ is connected and thus rational if and only if the polynomial possesses at most two real roots.  Then the result follows immediately from Theorem~\ref{thm birational to tubular}.

\end{proof}

%\begin{exmp}\todo{ten tu asi nakonec ani nebude nevim cio s nim...} {\color{blue} Jsem pro ho zachovat. JCAM vyzaduje pristup "cim vic prikladu, tim lip" :-)}
%\rm
%  Let $\av{X}$ be a cubic SOR. After a suitable isometry the coordinate $x$-axis is its axis as usual and hence its defining polynomial is $a(x)(y^2+z^2)-b(x)$, where
%  $\deg(a)\leq 1$, $\deg(b)\leq 3$ and $\gcd(a,b)=1$. Passing to the projective extension, this polynomial homogenizes as
%  \begin{equation}
%    A(W,X)(Y^2+Z^2)-B(W,X)
%  \end{equation}
%  The projective transformation $W'=L(W,X)$ $Y'=Y$, $Z'=Z$, $X'=??$ and the dehomogenization $x=X'/W'$, $y=Y'/W'$ and $z=Z'/W'$ leads to a new SOR $\av{X}'$ given by the equation
%  \begin{equation}
%    y^2+z^2-\beta(x-\alpha_1)(x-\alpha_2)(x-\alpha_3)=0
%  \end{equation}
%
%\end{exmp}

Finally, let us summarize all the obtained results on (uni)rationality of the surfaces of revolution reflecting the rationality of $\av{P}$ and $\av{P}^2$ to the following table, see Table~\ref{tab_SOR}.

\begin{table}[ht]
\centering
\begin{minipage}{0.9\textwidth}
\caption{Rationality and unirationality of SOR.}\label{tabulka}
\label{tab_SOR}
\end{minipage}\vspace{0.5ex}

        \begin{tabular}{|c|c|c|c|c} \hline
          $\mathrm{genus}(\av{P})$   & $\mathrm{genus}(\av{P}^2)$ & $\av{X}_{\cis}$ & $\av{X}$  \\ \hline\hline
          reducible          & $0$           & rational        & rational            \\ \hline
          reducible          & $>0$          & non-rational      & non-rational          \\ \hline
          $ 0$          & $0$           & rational        & depends on $\#\mathrm{components}$            \\ \hline
          $>0$          & $0$           & rational        & depends on $\#\mathrm{components}$     \\ \hline
          $>0$          & $>0$          & non-rational      & non-rational          \\  \hline
        \end{tabular}
\end{table}

%%%%%%%%%%%%%%%%%%%%%%%%%%%%%%%%%%%%%%%%%%%%%%%%%%%%%%%%%%%%%%%%%%%%%%%%%%%%%%%%%%%%%%%%%%%%%%%%%%%%%%%%%%%%%%%%%%%%%%%%%%%%%%%%%%%%%%%
\section{Conclusion}\label{Concl}
%%%%%%%%%%%%%%%%%%%%%%%%%%%%%%%%%%%%%%%%%%%%%%%%%%%%%%%%%%%%%%%%%%%%%%%%%%%%%%%%%%%%%%%%%%%%%%%%%%%%%%%%%%%%%%%%%%%%%%%%%%%%%%%%%%%%%%%

This paper was devoted to an interesting (and till now unsolved) theoretical problem, motivated by some technical applications, i.e., how to recognize an implicit surface of revolution from the defining polynomial equation of a given algebraic surface. We designed a symbolic algorithm (which avoids computing with float coefficients) returning for surfaces of revolution also their axis. In addition, we investigated the problem of rationality and unirationality of surfaces of revolution and presented how to solve this easily by discussing the rationality of a certain planar curve associated to the given rotational surface. The methods and approaches were presented on two examples in Appendix. The study can be considered as a first step towards the recognition of other implicitly given surfaces (e.g. canal surfaces, whose special instances surfaces of revolution are).

%%%%%%%%%%%%%%%%%%%%%%%%%%%%%%%%%%%%%%%%%%%%%%%%%%%%%%%%%%%%%%%%%%%%%%%%%%%%%%%%%%%%%%%%%%%%%%%%%%%%%%%%%%%%%%%%%%%%%%%%%%%%%%%%%%%%%%%
\section*{Acknowledgments}
%%%%%%%%%%%%%%%%%%%%%%%%%%%%%%%%%%%%%%%%%%%%%%%%%%%%%%%%%%%%%%%%%%%%%%%%%%%%%%%%%%%%%%%%%%%%%%%%%%%%%%%%%%%%%%%%%%%%%%%%%%%%%%%%%%%%%%%

The first author was supported by the project NEXLIZ, CZ.1.07/2.3.00/30.0038, which is co-financed by the European Social Fund and the state budget of the Czech Republic.
The work on this paper was supported by the European Regional Development Fund, project ``NTIS~--~New Technologies for the Information Society'', European Centre of Excellence, CZ.1.05/1.1.00/02.0090.
%We thank to all referees for their valuable comments, which helped us significantly to improve the paper.

\bigskip
\begin{appendix}

\section{Computed examples}
The methods and approaches studied in this paper will be now presented in detail in the two following particular examples.

\begin{exmp}\rm
Let $\av{X}$ be an implicit surface given by the defining polynomial
\begin{equation}
\begin{array}{rcl}
   f(x,y,z) & = & 729 x^6-5832 x^5 y+12150 x^5+19440 x^4 y^2-40500 x^4 y+6075 x^4 z^2-70750 x^4 \\
   && -34560 x^3 y^3-32400 x^3 yz^2+444000 x^3 y+67500 x^3 z^2+120000 x^3+34560 x^2 y^4 \\
   && +144000 x^2 y^3+64800 x^2 y^2 z^2-781750 x^2y^2-45000 x^2 y z^2-1555000 x^2 y \\
   && +16875 x^2 z^4-325000 x^2 z^2-189375 x^2-18432 x y^5-192000 x y^4-57600 xy^3 z^2\\
   && -240000 x y^2 z^2+2152500 x y^2-45000 x y z^4+1200000 x y z^2-1995000 x y+93750 xz^4\\
   && +156000 x y^3+675000 x z^2-3168750 x+4096 y^6+76800 y^5+19200 y^4 z^2+232375 y^4\\
   && +240000 y^3 z^2-390000 y^3+30000 y^2z^4+106250 y^2 z^2-388750 y^2+187500 y z^4\\
   && -1525000 y z^2+3287500 y+15625 z^6-406250 z^4+2265625 z^2-3562500.
\end{array}
\end{equation}
We choose randomly in $\euR{3}$ five points $\vek[p]_1,\ldots,\vek[p]_5$ (determining five associated surfaces $\av{X}_{f(\vek[p]_i)}$ from the family $\Sigma_f$), for instance
\begin{equation}
\{(-2, 1, 0), (0, 1, 0), (-2, -2, 1), (1, 1, -2), (-2, -2, 2)\}.
\end{equation}
Then we find the corresponding normals $N_{\vek[p]_1}\av{X}_{f(\vek[p]_1)}, \ldots,N_{\vek[p]_5}\av{X}_{f(\vek[p]_5)}$ and compute their Pl\"ucker coordinates
\begin{equation}
\begin{array}{rcl}
 N_1&=& (5114:-4452:0:0:0:3790), \\
 N_2&=& (-3065682:2678076:0:0:0:3065682), \\
 N_3&=& (1161776:9625632:11672400:-32970432:24506576:-16927712), \\
 N_4&=& (-797368:5955324:-8737800:3172848:10332536:6752692), \\
 N_5&=& (122126:1249332:3087300:-8673264:6418852:-2254412).
\end{array}
\end{equation}
It can be shown that these normals are linearly independent and thus the system of linear equations \eqref{eq_soustava} has only one homogeneous solution
\begin{equation}\label{eq exmp axis}
  A=(4: 3: 0:0: 0: -5),
\end{equation}
describing a unique line $\av{A}$ which can be parameterized as
\begin{equation}\label{eq exmp axis2}
  (3/5, -4/5, 0) +t(4, 3, 0).
\end{equation}

Now, we use an isometry $\phi$ which maps the axis $\av{A}$ to the coordinate $x$-axis and obtain the transformed surface $\av{X}'=\phi(\av{X})$ described by the polynomial
\begin{equation}
\begin{array}{rcl}
 \hat{f}(x,y,z) & = & -400 - 104 x^2 - x^4 + 200 y^2 - 48 x y^2 + 26 x^2 y^2 - 29 y^4 + 12 x y^4 + y^6 + 200 z^2 \\
 &&  - 48 x z^2 + 26 x^2 z^2 - 58 y^2 z^2 +  24 x y^2 z^2 + 3 y^4 z^2 - 29 z^4 + 12 x z^4 + 3 y^2 z^4 + z^6.
\end{array}
\end{equation}
The section with the coordinate plane $z=0$ has the equation
\begin{equation}\label{eq exmp profile}
 p(x,y)=-400 - 104 x^2 - x^4 + 200 y^2 - 48 x y^2 + 26 x^2 y^2 - 29 y^4 +
 12 x y^4 + y^6 = 0.
\end{equation}
It is seen that $p(x,y)$ contains $y$ in even powers only, and thus it is a profile curve of SOR $p(x,y^2+z^2)=0$. Finally, it can be easily verified that
$\hat{f}(x,y,z)=p(x,y^2+z^2)$. This brings us to the result that $\av{X}$ is SOR with axis~\eqref{eq exmp axis2}.
\end{exmp}

\bigskip
\begin{exmp}\rm
In the previous example, we have shown that $\av{X'}$ is SOR (with the axis of rotation being the coordinate $x$-axis). The genus of its profile curve $\av{P}$ given by the defining polynomial \eqref{eq exmp profile} is one and thus it is a non-rational curve. 

Hence to test the unirationality of $\av{X'}$ we have to use Theorem~\ref{thm rationality of sor} (see also Table~\ref{tabulka}).
Unlike $\av{P}$, the curve
\begin{equation}
\av{P}^2:400 - 104 x^2 - x^4 + 200 y - 48 x y + 26 x^2 y - 29 y^2 +  12 x y^2 + y^3=0
\end{equation}
is rational and it is parameterizable e.g. as
\begin{equation}
  \vek[p](t)=(-t^3+t,5 + 4 t + 6 t^2 + 4 t^3 + t^4).
\end{equation}
Now using the notation from Theorem~\ref{thm birational to tubular} we have
\begin{equation}
\tilde{p}(t)=-t^3+t,\quad \tilde{r}(t)=t^4+4t^3+6t^2+4t+5,\quad \tilde{q}(t)=1\quad \mbox{and}\quad d(t)=1.
\end{equation}
So $\av{X'}$ is birational to the tubular surface \eqref{eq tubular surface}
\begin{equation}
  \av{T}: y^2+z^2- t^4-4t^3-6t^2-4t-5=0
\end{equation}
via the mapping~\eqref{eq tau} $\tau:\av{T\dashrightarrow X'}$
\begin{equation}
  (x,y,z)\mapsto (-x^3+x,y,z).
\end{equation}
Since $\av{T}$ can be parameterized as
\begin{equation}
  \left(t,
  \frac{2 s \left(t^2+2 t-1\right)+\left(s^2-1\right) (2
   t+2)}{s^2+1},
   \frac{\left(1-s^2\right) \left(t^2+2 t-1\right)+2 s (2 t+2)}{s^2+1}\right),
\end{equation}
we arrive at a parameterization of $\av{X'}$ in the form
\begin{equation}
  \left(-t^3+t,\frac{2 s \left(t^2+2 t-1\right)+\left(s^2-1\right) (2
   t+2)}{s^2+1},
   \frac{\left(1-s^2\right) \left(t^2+2 t-1\right)+2 s (2 t+2)}{s^2+1}
  \right).
\end{equation}
Finally the transformation $\phi^{-1}$ (see the previous example) leads to a rational parameterization  of $\av{X}=\phi^{-1}(\av{X}')$.

\end{exmp}
\end{appendix}

\bibliographystyle{ieeetr}
\bibliography{bibliography}{}

\end{document}